\documentclass[]{article}
\usepackage{amsmath,amsfonts,amssymb}
\usepackage{amsthm}
\usepackage[mathfrak]{}
\usepackage{cases}
\usepackage[latin1]{inputenc}
\usepackage[T1]{fontenc}
\usepackage{verbatim}
\usepackage{graphicx}
\usepackage{float}
\usepackage{mathrsfs}
\usepackage [all]{xy}
\usepackage{color}
\usepackage{tikz}
\usepackage{comment}
\usepackage{hyperref}
\usepackage{cases}

\newtheorem{Theorem}{Theorem}[section]
\newtheorem{Lemma}[Theorem]{Lemma}
\newtheorem{Corollary}[Theorem]{Corollary}
\newtheorem{Proposition}[Theorem]{Proposition}
\newtheorem{Definition}[Theorem]{Definition}
\newtheorem{Example}[Theorem]{Example}

\title{Error-Correcting Codes on Projective Bundles over Deligne--Lusztig Varieties}
\author{Daniel Camaz\'on \footnote{Department of Algebra, Analysis, Geometry and Topology of the University of Valladolid} \footnote{The first author was partially supported by PGC2018-096446-B-C21}, Juan Antonio L\'opez Ramos \footnote{Department of Mathematics of the University of Almer\'ia} \footnote{The second suthor is supported by FQM 0211 Junta de Andaluc\'ia and Ministerio de Ciencia e Innovaci\'on PID2020-113552GB-I00}}
\date{daniel.camazon@uva.es, jlopez@ual.es}

\begin{document}
\maketitle

\begin{abstract}
The aim of this article is to give lower bounds on the parameters of algebraic geometric error-correcting codes constructed
from projective bundles over Deligne--Lusztig surfaces. The methods based on an intensive use of the intersection theory allow us to extend the codes previously constructed from higher-dimensional varieties, as well as those coming from curves. General bounds are obtained for the case of projective bundles of rank $2$ over standard Deligne-Lusztig surfaces, and some explicit examples coming from surfaces of type $A_{2}$ and ${}^{2}A_{4}$ are given.
\end{abstract}

\section{Introduction}

The theory of algebraic geometry codes arose in 1970, when V. Goppa discovered in \cite{goppa}
the relation between the theory of error-correcting codes and the evaluation on algebraic curves. They exhibited important properties that made researchers deepen their study. On the one hand, their good encoding--decoding algorithms led B. McEliece
(\cite{mceliece}) to consider them for a public-key cryptosystem, which is considered by NIST as an alternative for the post-quantum era. On the other hand, M.A. Tsfasman et al. (\cite{tsfasman}) were able to show that Goppa codes can be used to give examples of codes that go beyond the Gilbert--Varshamov bound. Since then, the interest in algebraic geometry has significantly increased.  We can cite \cite{barelli,christensen,jphansen,xing,munuera,stichtenorth} as a few of the many existing examples of study over Hermitian, Castle, Suzuki or G-H curves.

But the definition of algebraic geometry codes can go beyond. Tsfasman and Vl\v{a}du\k{t}, in \cite{tsfasman2}, suggested that higher-dimensional varieties can be used to construct these type of codes, although the number of works in this sense does not equal that of the curves, probably due to the difficulty of finding higher-dimensional varieties $X$, spaces of functions $L$ and sets of rational points ${\cal P}$, that can yield codes in the sense of Definition \ref{DefAGCodes}, which could be interesting due to their compelling properties concerning their weight distributions, minimum distance or fast encoding--decoding algorithms.
 The two-dimensional case has been relatively studied, like in the case of rational, hermitian or cubic surfaces, (cf. \cite{couvreur,edoukou,voloch}).

However, this is not the case in general for higher-dimensional varieties. The survey by J. Little (\cite{little}) offers a rather complete vision on the study of algebraic geometry codes defined over varieties in general. Among the referred papers, we can find the work by S. Hansen \cite{shhansen}, wherein the author makes an extensive use of the intersection theory to develop their study. In this paper, some of the provided examples concern Deligne--Lusztig varieties, whose importance in algebraic geometry comes from the fact that they are directly linked to finite groups of Lie type. Moreover, these varieties are characterized by their large number of rational points, which allows for the definition of algebraic geometry codes, as we will see. In his paper, S. Hansen applies general methods to obtain lower bounds on the parameters of algebraic geometric error-correcting codes defined from varieties of greater dimension, as is the case of Deligne--Lusztig surfaces. Our aim in this paper is to go one step further \cite{shhansen} and study this kind of codes defined on projective bundles over Deligne--Lusztig surfaces. In order to obtain lower bounds on the associated parameters, we make an intensive use of the intersection theory and take advantage of the fact that, for some standard Deligne--Lusztig surfaces, all their rational points are distributed equally on the disjoint rational curves, constituting the irreducible components of a divisor $D_{i}$. This allows us to give general bounds for the case of algebraic geometric error-correcting codes on projective bundles of rank $2$ as well as some explicit examples coming from surfaces of type $A_{2}$ and ${}^{2}A_{4}$.

\section{Deligne--Lusztig Varieties}

In this section, we shall define Deligne--Lusztig varieties and their compactification as well as study some of their main properties.

Let $(G,F)$ be a connected reductive algebraic group over an algebraically closed field $k$ of positive characteristic $p$, equipped with an $\mathbb{F}_{q}-$structure coming from a Frobenius morphism $F: G\rightarrow G$. Let $L: G\rightarrow G$ be the corresponding Lang map taking an element $g\in G$ to $g^{-1}F(g)$. By the Lang--Steinberg Theorem (see Theorem 4.4.17 in \cite{Springer98}), 
 \color{black} this morphism of varieties is surjective with finite fibers. From this result, it follows that, by conjugacy of Borel subgroups,
there exists an $F-$stable Borel subgroup $B$. Let $\pi: G\rightarrow G/B:= X$ denote the quotient. There are then (with a slight abuse of notation) natural endomorphisms $F: W\rightarrow W$ and $F: X\rightarrow X$ of the Weyl group of $G$ and the variety $X$ of Borel subgroups of $G$. Let $W$ be generated by the simple reflections $s_{1}, \dots, s_{n}$, and let $l(\cdot)$ be the length function with respect to these generators.

 Let us now recall from Definition 1 in \cite{Hansen02} the following definition of Deligne--Lusztig variety.

\color{black}
\begin{Definition}\label{DefDelLusVar}
Fix an element $w$ in the Weyl group $W$, and let $w=s_{i_{1}}\cdot \ldots \cdot s_{i_{r}}$ be a reduced expression of $w$. Call $w$ a Coxeter element if there in this expression occurs exactly one $s_{i}$ from each of the orbits of $F$ on $\left\{s_{1}, \dots, s_{n}\right\}$. Denote by $\delta$ the order of $F$ on this set. Then, the Deligne--Lusztig variety $X(w)$ is defined as the image of $L^{-1}(B\dot{w}B)$ in $G/B$. That is,
\begin{equation}
X(w)=\pi(L^{-1}(B\dot{w}B)).
\end{equation}
\end{Definition}

 Next, we will follow the notation and definitions given in \cite{Hansen02}. \color{black} Define the closed subvariety of $X^{r+1}$
\vspace{-12pt}

\begin{equation}
\overline{X}(s_{i_{1}}, \dots, s_{i_{r}})=\left\{(g_{0}B, \dots, g_{r}B)\in X^{r+1}: g_{k}^{-1}g_{k+1}\in B\cup Bs_{i_{k+1}}B\enspace\text{for}\enspace 0\leq k<r , g_{r}^{-1}F(g_{0})\in B\right\}
\end{equation}

\noindent In 
 those cases wherein there is a unique product $s_{i_{1}}\cdot\ldots\cdot s_{i_{r}}$ such that $s_{i_{1}}\cdot\ldots\cdot s_{i_{r}}=w$, we
shall write $\overline{X}(w)$ for the variety $\overline{X}(s_{i_{1}}, \ldots, s_{i_{r}})$.
For any subset $\left\{s_{j_{1}}, \ldots, s_{j_{m}}\right\}\subset\left\{s_{i_{1}}, \ldots, s_{i_{r}}\right\}$, $\overline{X}(s_{j_{1}}, \ldots, s_{j_{m}})$ defines in a natural way a closed subvariety of $\overline{X}(s_{i_{1}}, \ldots, s_{i_{r}})$. In particular, there are divisors
\begin{equation}
D_{j}=\overline{X}(s_{i_{1}}, \ldots, \hat{s}_{i_{j}}, \ldots, s_{i_{r}});\enspace j=1, \ldots, r.
\end{equation}

\begin{Example}
Let us consider the Deligne--Lusztig surface $\overline{X}(w)$ of type $A_{2}$. In this particular case, the connected reductive algebraic group is $G=GL(3,k)$ and the Frobenius map is given by
\begin{equation}
\xymatrix{F: g\ar[r] & g^{q}}.
\end{equation}

\noindent Moreover, the Weyl group $W$ of $G$ is isomorphic to the symmetric group of the three vectors in the base of $k^{3}$, and $\overline{X}(w)=\overline{X}(s_{1},s_{2})$, where $s_{1}$ corresponds to the permutation of the first and the second vectors of the base and $s_{2}$ corresponds to the permutation of the second and the third vector.
\end{Example}
\color{black}

When $G$ is semi-simple with connected Dynkin diagram $D$ (with numbering of nodes
and their associated simple reflections), there is a (unique) natural
choice of Coxeter element: let $w=s_{1}\cdot s_{2}\cdot\ldots\cdot s_{r}$ with $r$ maximal (under the condition
that $s_{r}$ is not in the $F-$orbit of any of the previous $s_{i}$, $i<r$). When choosing this particular Coxeter element, we shall refer to $X(w)$ (or $\overline{X}(w)$) as being of standard type. \\
We claim that $\overline{X}(w)$ is of classical type if $w$ is a Coxeter element for one of the following classical groups: $A_{n}, {}^{2}A_{2n}, {}^{2}A_{2n+1}, B_{m}, C_{n}, D_{n}$ or ${}^{2}D_{n}$. \\
For $w_{1},w_{2}\in W$, we shall say that $w_{1}$ and $w_{2}$ are $F-$conjugate if there exists $w{'}\in W$ such
that $w_{2}=w{'}w_{1}F(w{'})^{-1}$. It is worth noting that $w$ and $F(w)$ are $F-$conjugate for any $w\in W$.

Since the morphism $L$ is flat, it is open; hence, $\overline{L^{-1}(B\dot{w}B)}=L^{-1}(\overline{B\dot{w}B})$. Therefore, $X(w)$ is a non-singular variety of dimension $n$ and the closure of $X(w)$ in $X$, $\overline{X(w)}$, is given by the disjoint union
\begin{equation}
\overline{X(w)}=\bigcup_{w{'}\leq w} X(w{'}),
\end{equation}
where, as usual, $\leq$ is the Bruhat order in $W$. This closure is usually singular whenever the
Schubert variety $X_{w}=\overline{B\dot{w}B}/B$ is. But since the open subset
\begin{equation}
\left\{g_{0}B, \ldots, g_{r}B)\in X^{r+1}: g^{-1}_{k}g_{k+1}\in Bs_{i_{k}+1}B, 0\leq k<r, g^{-1}_{r}F(g_{0})\in B\right\}
\end{equation}	
of the smooth projective variety $\overline{X}(w)$ maps isomorphically onto $X(w)$ under projection to
the first factor, we have a good compactification of $X(w)$. In fact, the complement
of $X(w)$ in $\overline{X}(w)$, which is easily seen to be the union of the divisors $D_{j}$ defined above, is a
divisor with normal crossings.
If $w$ is a Coxeter element, then $X(w)$ and $\overline{X}(w)$ are irreducible and, in fact, $X(w)$ is isomorphic to $\overline{X}(w)$, and hence non-singular.

 It is worth noting that it follows from Definition \ref{DefDelLusVar} that, if $\overline{X}(w)$ is of classical type, then the irreducible components
of any Deligne--Lusztig subvariety of $\overline{X}(w)$ are of classical type too (see Remark 1 in \cite{Hansen02}).

\color{black}
The following result allows us to consider the image of standard Deligne--Lusztig varieties under a certain proper morphism as a normal strict complete intersection in a certain projective space $\mathbb{P}^{N-1}$ as well as to compute its Picard group.

\begin{Theorem}{\emph{Theorem 3 in \cite{Hansen02}}}\label{ThmHyp}
Let $\overline{X}(w)$ be a standard Deligne--Lusztig variety of type ${}^{2}A_{n}$, $B_{n}$, $C_{n}$, $D_{n}$ or ${}^{2}D_{n}$. Assume that $char(k)\neq 2$ in the orthogonal cases. Let $P$  be the parabolic subgroup generated by $B$ together with the double cosets $Bs_{2}B, Bs_{3}B, \ldots,Bs_{n}B$, and let
\begin{equation}
\pi: (G/B)^{l(w)+1}\rightarrow G/P\subseteq\mathbb{P}(V)\cong\mathbb{P}^{N-1}
\end{equation}
be the projection (projection to the first factor, followed by the quotient map). It is worth noting that the inclusion $G/P\subseteq\mathbb{P}(V)$ is an
equality in the non-orthogonal cases. Denote by $L^{e}$ the $e-$dimensional linear subspace of $\mathbb{P}^{N-1}$ obtained by setting the $N-1-e$ last coordinates as equal to zero.
\begin{enumerate}
\item The image $Z=\pi(\overline{X}(w))$ is a normal, strict complete intersection. In fact, Lusztig shows (see pp. 444--445 in \cite{Lusztig76}) that $Z$ equals the support of the scheme's theoretic complete intersection $Z{'}= \cap_{i=0}^{a_{0}(V)-1}H_{i}$, with $X(w)$ mapping isomorphically onto the open subset $\left\langle x,F_{V}^{a_{0}(V)}(x)\right\rangle\neq 0$ of $Z$  (see Table 2 in \cite{Hansen02}). \color{black} In the unitary and orthogonal cases, the singular locus of $Z$, $Z_{sing}$ consists of the finitely many $G^{F}-$translates of the closed subscheme $Z\cap L^{a_{0}(V)-1}$. Hence,
\begin{equation}
codim(Z_{sing},Z)=N+1-2a_{0}(V)+a_{0}(L^{a_{0}(V)-1}).
\end{equation}

\noindent In the symplectic case, $Z_{sing}$ consists of the $G^{F}-$translates of the closed subscheme $Z\cap L^{a_{0}(V)-2}$, and the previous formula becomes
\begin{equation}
codim(Z_{sing},Z)=2+a_{0}(L^{a_{0}(V)-2}).
\end{equation}
\item For $codim(Z_{sing},Z)\geq 4$, $Pic(Z)=\mathbb{Z}$ and consequently
\begin{equation}
\begin{split}
Pic(\overline{X}(w))=\mathbb{Z}\left[\pi^{*}H\right] & \oplus\mathbb{Z}\left[\left\{\left[V\right]:\enspace V\enspace\text{component of}\enspace D_{1}\right\}\right]\\
& \oplus j_{*} A_{l(w)-1}(\bigcup_{i\in\mathcal{I}-\left\{1\right\}}D_{i})
\end{split}
\end{equation}
where $H$ is the hyperplane section of $Z$, $j$ is the obvious inclusion and $\mathcal{I}$ is the set of indices $\left\{i\right\}$, satisfying that some connected component of the Dynkin diagram corresponding to $D_{i}$ occurs as a subgraph of the Dynkin diagram corresponding to $D_{1}$.
\item For any Coxeter element $w{'}$, we have
\begin{equation*}
Pic(\overline{X}(w{'}))_{p{'}}\cong Pic(\overline{X}(w))_{p{'}}.
\end{equation*}
\end{enumerate}
\vspace{-12pt}
\end{Theorem}

\section{Error-Correcting Codes Construction}

Tsfasman and Vl\v{a}du\k{t} introduced the following construction (generalizing the Goppa--Manin construction):

\begin{Definition}{\emph{H-construction p. 272 in \cite{tsfasman2}}}\label{DefAGCodes}.
\color{black} Let $X$ be a normal projective variety over $\mathbb{F}_{q}$. Let $L$ be a line bundle defined over
$\mathbb{F}_{q}$ and let $P_{1}, P_{2}, ...P_{n}$ be distinct $\mathbb{F}_{q}-$rational points on $X$. Set $\mathcal{P}=\left\{P_{1}, P_{2}, ...P_{n}\right\}$. In each $P_{i}$ , choose isomorphisms of the fibers $L_{P_{i}}$ with $\mathbb{F}_{q}$. The linear code $C(L,\mathcal{P})$ of length $n$ associated to $(X,L,\mathcal{P})$ is the image of the germ map
\begin{equation}
\alpha: \Gamma(X,L)\rightarrow \oplus_{i=1}^{n} L_{P_{i}}\cong\mathbb{F}_{q}.
\end{equation}

\noindent From now on, we assume that all line bundles considered
actually have a non-zero global section. \\
Suppose $L$ arises as the line bundle associated to a divisor $D$ and that the $P_{i}$
are not in the support of $D$. Then, we obtain the same code (up to isomorphism) as
when evaluating the rational functions
\begin{equation*}
L(D)=\left\{f\in k(X)^{*}: div(f)+D\geq 0\right\}
\end{equation*}
in the points $\mathcal{P}$.
\end{Definition}

The fundamental question is: Given a line bundle $L$ on $X$, how many zeros does a section $s\in\Gamma(X,L)$
have along a fixed set $\mathcal{P}$ of rational points? \\
Using the correspondence between line bundles and (Weil) divisors on a normal variety, we may reformulate the question as follows:
For a fixed line bundle $L$, and given an effective divisor $D$ such that $L=\mathcal{O}_{X}(D)$, how many points from $\mathcal{P}$ are in its support $\left|D\right|$?

Although in the particular case of $dim(X)=1$, where the points $P\in\mathcal{P}$ happen to be divisors, one may apply the Riemann--Roch theorem to give a lower bound on $d$ and a formula for $k$ in higher dimensions; however, we have to face the task of comparing objects of different dimension. This may be remedied in two ways:
\begin{enumerate}
\item Make the objects have the same codimension by blowing-up at the points;
\item Make the objects have complementary dimensions, that is, make the points in some way into curves.
\end{enumerate}

In the next section, we shall pursue the latter idea. In this case, the following result establishes a lower bound for the minimum distance.

\begin{Proposition}{\emph{Proposition 3.2 in \cite{shhansen}}}\label{ProMDEst}
Let $X$ be a normal projective variety defined over $\mathbb{F}_{q}$ of
dimension at least two. Let $C_{1}, C_{2}, ..., C_{a}$ be (irreducible) curves on $X$ with $\mathbb{F}_{q}-$
rational points $\mathcal{P}=\left\{P_{1}, P_{2}, \ldots, P_{n}\right\}$. Assume the number of $\mathbb{F}_{q}-$rational points on
rational $C_{i}$ is less than $N$. Let $L$ be a line bundle on $X$, defined over $\mathbb{F}_{q}$, such
that $L\cdot C_{i}\geq 0$ for all $i$. Let
\begin{equation*}
l=sup_{s\in\Gamma(X,L)}\#\left\{i:Z(s)\enspace\text{contains}\enspace C_{i}\right\}.
\end{equation*}

\noindent Then, the code $C(L,\mathcal{P})$ has length $n$ and minimum distance
\begin{equation*}
d\geq n-lN-\sum_{i=1}^{a}L\cdot C_{i}
\end{equation*}

\noindent If $L\cdot C_{i}=\eta\leq N$ for all $i$, then $d\geq n-lN-(a-l)\eta$.
\end{Proposition}

In particular, if $X$ is a non-singular surface, we can cite the following corollary.

\begin{Corollary}{\emph{Corollary 3.2 in \cite{shhansen}}}\label{CorMDNefD}.
Assume furthermore that $X$ is a non-singular surface and
that $H$ is a nef divisor on $X$ with $H\cdot C_{i}>0$ for all $i$. Then,
\begin{equation}
\mathit{l}\leq\frac{L\cdot H}{min_{i}\left\{C_{i}\cdot H\right\}}
\end{equation}

\noindent Consequently, if $L\cdot H<C_{i}\cdot H$ for all $i$, we have $\mathit{l}=0$ and
\begin{equation*}
d\geq n-m,
\end{equation*}
where $m=\sum_{i=1}^{a}L\cdot C_{i}$.
\end{Corollary}

\section{Some EC-Codes on Projective Bundles over Standard Deligne--Lusztig Surfaces}

This section is devoted to the computation of lower bounds for the parameters of certain error-correcting codes on projective bundles over standard Deligne--Lusztig surfaces.

\medskip

Given a standard Deligne--Lusztig surface $S$ of type $A_{2}(\mathbb{F}_{q}), {}^{2}A_{3}(\mathbb{F}_{q^{2}}), {}^{2}A_{4}(\mathbb{F}_{q^{2}})$ or $C_{2}(\mathbb{F}_{q})$, for a suitable parabolic subgroup $P$ of $G$ we have a commutative diagram:
\begin{equation}\label{Diag1}
\xymatrix{ S\ar[r]^{j}\ar[d]_{\rho} & G/B\ar[d]^{\pi} \\
           Z\ar[r]^{i} & G/P\cong\mathbb{P}^{N-1}}
\end{equation}
where 
\begin{align}
Z= & \begin{cases} \mathbb{P}^{2} & \text{if $S=A_{2}(\mathbb{F}_{q})$}, \\
                   H_{0}\equiv X_{0}^{q+1}+X_{1}^{q+1}+X_{2}^{q+1}+X_{3}^{q+1}=0 & \text{if $S={}^{2}A_{3}(\mathbb{F}_{q^{2}})$}, \\
									 H_{0}\equiv X_{0}^{q}X_{3}-X_{0}X_{3}^{q}+X_{1}X_{2}^{q}-X_{1}^{q}X_{2}=0 & \text{if $S=C_{2}(\mathbb{F}_{q})$}, \\
									 \end{cases}
\end{align}
and $Z=H_{0}\cap H_{1}$ with
\begin{align}
H_{0} & \equiv X_{0}^{q+1}+X_{1}^{q+1}+X_{2}^{q+1}+X_{3}^{q+1}+X_{4}^{q+1}=0, \\
H_{1} & \equiv X_{0}^{q^{3}+1}+X_{1}^{q^{3}+1}+X_{2}^{q^{3}+1}+X_{3}^{q^{3}+1}+X_{4}^{q^{3}+1}=0,
\end{align}
if $S={}^{2}A_{4}(\mathbb{F}_{q^{2}})$, $i$ is an embedding, $j$ is finite, $\pi$ is locally trivial (in the Zariski topology) and $\rho$ is birational and surjective; see Sect. 5,6,7,8 in \cite{Rodier00}. 

Moreover, the surface $S$ is isomorphic to the blow-up of  (see Definition p.163 in \cite{Hartshorne77}) \color{black} $Z$ at a certain set of points in the cases $S=A_{2}(\mathbb{F}_{q}), {}^{2}A_{3}(\mathbb{F}_{q^{2}}), C_{2}(\mathbb{F}_{q})$, and all its rational points are distributed equally on the disjoint rational curves $B_{i}$ constituting the irreducible components of the divisor $D_{1}\subset S$, whereas in the remaining case, $S={}^{2}A_{4}$, $Z$ is obtained from $S$ by contracting the disjoint hermitian curves $A_{i}$ constituting the irreducible components of the divisor $D_{1}\subset S$ and all the rational points of $S$ are distributed equally on the disjoint rational curves $B_{i}$ constituting the irreducible components of the divisor $D_{2}\subset S$.

The following theorem constitutes the main result of this paper, where explicit lower bounds for the parameters of a certain class of codes on projective bundles over Deligne--Lusztig surfaces are given.

\begin{Theorem}\label{ThmStdDLSur}
Let $S$ be a standard Deligne--Lusztig surface of type $A_{2}(\mathbb{F}_{q}), {}^{2}A_{3}(\mathbb{F}_{q^{2}}), {}^{2}A_{4}(\mathbb{F}_{q^{2}})$ or $C_{2}(\mathbb{F}_{q})$, and let $V$ be a vector bundle of rank $2$ defined over $S$. Consider the projective bundle $T=P(V)$ over $S$, that is $p: T\rightarrow S$. Then, for some $a,b>0$, we can construct a code on $T$ over $\mathbb{F}_{q^{\delta}}$ with parameters
\begin{enumerate}
\item $n=\#S(\mathbb{F}_{q^{\delta}})\#\mathbb{P}^{1}(\mathbb{F}_{q^{\delta}})$;
\item $k=h^{0}(S,Symm^{b}(V)\otimes\mathcal{O}_{S}(aD_{j}))$;
\item $d\geq n-(-bc_{1}(W_{1})\#\left\{B_{i}\right\}+aD_{j}\cdot D_{i})\#\mathbb{P}^{1}(\mathbb{F}_{q^{\delta}})-(\#S(\mathbb{F}_{q^{\delta}})-(-bc_{1}(W_{1})\#\left\{B_{i}\right\}+aD_{j}\cdot D_{i}))b$
if $(-bc_{1}(W_{1})\#\left\{B_{i}\right\}+aD_{j}\cdot D_{i})>0$, and $d\geq n-(\#S(\mathbb{F}_{q^{\delta}}))b$ otherwise,
\end{enumerate}
where
\begin{align*} 
\delta= & \begin{cases}
1 & \text{if $S$ is of type $A_{2}, C_{2}$,} \\
2 & \text{if $S$ is of type ${}^{2}A_{3}, {}^{2}A_{4}$,} \\
\end{cases} & D_{i}= & \begin{cases}
D_{1} & \text{if $S$ is of type $A_{2}, {}^{2}A_{3}, C_{2}$,} \\
D_{2} & \text{if $S$ is of type ${}^{2}A_{4}$,} \\
\end{cases} \\
D_{j}= & \begin{cases}
D_{2} & \text{if $S$ is of type $A_{2}, {}^{2}A_{3}, C_{2}$,} \\
D_{1} & \text{if $S$ is of type ${}^{2}A_{4}$,} \\
\end{cases} & & 
\end{align*}
and $c_{1}(W_{1})$ denotes the first Chern class of the line subbundle of minimum degree of the restricted vector bundle $i_{B_{i}}^{*} V$ over $B_{i}$.
\end{Theorem}

\begin{proof}
We will consider $\mathcal{P}$ to be the $\mathbb{F}_{q^{\delta}}-$rational points on $T$. Let $C_{1}, C_{2}, ..., C_{a}$ be the fibers over the $\mathbb{F}_{q^{\delta}}-$rational points of $S$. These disjoint lines contain all $\mathbb{F}_{q^{\delta}}-$rational points of $T$, that is,
\begin{equation}\label{EqNRPPB}
T(\mathbb{F}_{q^{\delta}})=\bigcup_{P\in S(\mathbb{F}_{q^{\delta}})} p^{-1}(P)(\mathbb{F}_{q^{\delta}}).
\end{equation}

\noindent It
 follows then that the length $n$ of the code is
\begin{equation}
n=\#\mathcal{P}=\#S(\mathbb{F}_{q^{\delta}})\cdot\#\mathbb{P}^{1}(\mathbb{F}_{q^{\delta}})
\end{equation}

\noindent Let $L$ be the line bundle $L=\mathcal{O}_{T}(b)\otimes\mathcal{O}_{T}(ap^{*}(D_{j}))$ over $T$. From Theorem 9.6 in \cite{EisenbudHarris16} and Proposition II.7.11 in \cite{Hartshorne77}, it follows that
\begin{equation}
\Gamma(T,L)\cong\Gamma(S,p_{*}L)=\Gamma(S,Symm^{b}(V)\otimes\mathcal{O}_{S}(aD_{j})),
\end{equation}
so when in the range wherein the bound on the minimum distance ensures the injectivity of the evaluation map, the dimension of the code is
\begin{equation}
k=h^{0}(S,Symm^{b}(V)\otimes\mathcal{O}_{S}(aD_{j})).
\end{equation}

\noindent Now, we will apply Proposition \ref{ProMDEst} in order to obtain the bound for the minimum distance. It is worth noting that, in the cases we are interested in, $C_{i}=p^{-1}(P)$ with $P\in S(\mathbb{F}_{q^{\delta}})$, so the maximum number of rational points on $C_{i}$ will be $N=\#\mathbb{P}^{1}(\mathbb{F}_{q^{\delta}})$. Moreover, by Lemma 9.7 in \cite{EisenbudHarris16}, we have that
\begin{equation}
L\cdot C_{i}=b.
\end{equation}

\noindent From Equation \ref{EqNRPPB}, we know that $\bigcup_{i=1}^{\#S(\mathbb{F}_{q^{\delta}})} C_{i}\subset p^{*}D_{i}$. Now, $B_{i}$ is an irreducible component of $D_{i}$, it is rational and $T_{B_{i}}=p^{-1}(B_{i})$ will be isomorphic to the projective bundle $p_{T_{B_{i}}}: P(i_{B_{i}}^{*} V)\rightarrow B_{i}$, where $i_{B_{i}}: B_{i}\rightarrow S$ denotes the closed embedding. As a consequence, by Theorem 9.6 in \cite{EisenbudHarris16}, the Chow ring of $T_{B_{i}}$ is isomorphic to 
\begin{equation}
A^{\bullet}(T_{B_{i}})\cong A^{\bullet}(B_{i})\left[\xi\right]/(\xi^{2}+c_{1}(i_{B_{i}}^{*} V)\xi)
\end{equation}
where $\xi$ is the hyperplane section in $T_{B_{i}}$. It is worth noting that since $B_{i}$ is rational, then by \linebreak Corollary V.2.14 in \cite{Hartshorne77}, $i_{B_{i}}^{*} V=W_{1}\oplus W_{2}$, and we can always suppose that $deg(W_{1})\leq deg(W_{2})$.\\

\noindent Now, if we restrict the line bundle $L$ to $T_{B_{i}}$, then $L\vert_{T_{B_{i}}}\cong\mathcal{O}_{T_{B_{i}}}(b)\otimes\mathcal{O}_{T_{B_{i}}}(p_{B_{i}}^{*}i_{B_{i}}^{*}(aD_{j}))$. Let $H=\xi+c_{1}(W_{2})F$ be a nef divisor on $T_{B_{i}}$ (see Theorem V.2.17 \cite{Hartshorne77}). Then, since $H\cdot C_{i}=1$, it follows by Corollary \ref{CorMDNefD} that
\begin{equation}
\mathit{l}_{T_{B_{i}}}\leq H\cdot L\vert_{B_{i}}=-bc_{1}(W_{1})+aD_{j}\cdot B_{i},
\end{equation}
provided that $(-bc_{1}(W_{1})+aD_{j}\cdot B_{i})>0$, and $\mathit{l}_{T_{B_{i}}}=0$ otherwise.
As $D_{i}=\sqcup B_{i}$, we can conclude the following bound for the minimum distance
\vspace{-12pt}

\begin{equation}
d\geq n-(-bc_{1}(W_{1})\#\left\{B_{i}\right\}+aD_{j}\cdot D_{i})\#\mathbb{P}^{1}(\mathbb{F}_{q^{\delta}})-(\#S(\mathbb{F}_{q^{\delta}})-(-bc_{1}(W_{1})\#\left\{B_{i}\right\}+aD_{j}\cdot D_{i}))b
\end{equation}
if $(-bc_{1}(W_{1})\#\left\{B_{i}\right\}+aD_{j}\cdot D_{i})>0$ and
\begin{equation}
d\geq n-(\#S(\mathbb{F}_{q^{\delta}}))b
\end{equation}
otherwise.
\end{proof}

Before computing the explicit parameters of some families of these codes, we need the following auxiliary result.

\begin{Lemma}\label{LemSADSLB}
Let $Y$ be a variety and $V=V_{1}\oplus V_{2}$ a vector bundle of rank $2$ over $Y$, that is, a direct sum of two line bundles $V_{1}$ and $V_{2}$. Then, the symmetric algebra of $V$ satisfies
\begin{equation}
Symm(V)\cong\bigotimes_{i=1}^{2} Symm(V_{i})=\bigoplus_{i_{1},i_{2} \in \mathbb{N}^{2}} Symm^{i_{1}}V_{1} \otimes Symm^{i_{2}}V_{2}\cong \bigoplus_{i_{1},i_{2} \in \mathbb{N}^{2}} V_{1}^{\otimes i_{1}} \otimes V_{2}^{\otimes i_{2}},
\end{equation}
\end{Lemma}

\begin{proof}
The first isomorphism is a consequence of the universal property for symmetric algebras, the second equality is just by definition and the last isomorphism follows from the fact that, for a line bundle $V_{i}$, we have
\begin{equation} 
V_{i}^{\otimes n}\cong Symm^{n}V_{i}.
\end{equation}
\end{proof}

 The following corollaries give explicit bounds for the parameters of certain families of codes on projective bundles of rank $2$ over Deligne-Lusztig surfaces of type $A_{2}$ and ${}^{2}A_{4}$, obtained by restricting the previous results to certain vector bundles defined over them.

\color{black}
\begin{Corollary}\label{CorCodA2}
Let $S_{1}$ be the Deligne--Lusztig surface of type $A_{2}$ defined over the field $\mathbb{F}_{q}$. Consider some $b$, such that $0<b<(q+1)$, and $V_{i}=\mathcal{O}_{S_{1}}(n_{i}H-\sum_{j=1}^{q^{2}+q+1} m_{i,j}B_{j})$ for $i=1,2$, where $H=\pi^{*}(\mathcal{O}_{\mathbb{P}^{2}}(1))$, verifying for any pair $i_{1},i_{2}\in\mathbb{N}$ with $i_{1}+i_{2}=b$:

\begin{enumerate}
\item $(i_{1}n_{1}+i_{2}n_{2})\leq 3(q-1)$;
\item $(i_{1}n_{1}+i_{2}n_{2})\geq\sum_{j=1}^{3}(i_{1}m_{1,j}+i_{2}m_{2,j})$;
\item $(i_{1}m_{1,j}+i_{2}m_{2,j})\geq (i_{1}m_{1,j+1}+i_{2}m_{2,j+1})$;
\item  $3(i_{1}n_{1}+i_{2}n_{2})>\sum_{j=1}^{q^{2}+q+1}(i_{1}m_{1,j+1}+i_{2}m_{2,j+1})$. 
\end{enumerate}
\color{black}

\noindent Let $T_{1}$ be the projective bundle $P(V_{1}\oplus V_{2})$ over $S_{1}$, $p_{1}: T_{1}\rightarrow S_{1}$. Then, we can construct a code on $T_{1}$ over $\mathbb{F}_{q}$ with parameters
\begin{enumerate}
\item $n=(q^{2}+q+1)(q+1)^{2},$;
\item $k=\sum_{\substack{i_{1},i_{2}\in\mathbb{N}^{2} \\ i_{1}+i_{2}=b}}\frac{1}{2}((i_{1}n_{1}+i_{2}n_{2})((i_{1}n_{1}+i_{2}n_{2})+3)-\sum_{j=1}^{q^{2}+q+1} (i_{1}m_{1,j}+i_{2}m_{2,j})((i_{1}m_{1,j}+i_{2}m_{2,j})+1))+1,$;
\item $d\geq n-(q^{2}+q+1)(q+1)b.$.
\end{enumerate}
\end{Corollary}

\begin{proof}
Since $A_{2}$ is isomorphic to the blow-up of $\mathbb{P}^{2}$ at the set of its rational points over $\mathbb{F}_{q}$ (see Sect. 5 in \cite{Rodier00}), we have that $\#S_{1}(\mathbb{F}_{q})=(q^{2}+q+1)(q+1)$. Moreover, $\#\mathbb{P}^{1}(\mathbb{F}_{q})=(q+1)$, so we can conclude that the length $n$ of the code is
\begin{equation}
n=\#\mathcal{P}=\#S_{1}(\mathbb{F}_{q})(q+1)=(q^{2}+q+1)(q+1)^{2}
\end{equation}

\noindent Let $L$ be the line bundle $L=\mathcal{O}_{T_{1}}(b)$. In order to compute the dimension of the code, we have $dim \Gamma(T_{1},L)=dim \Gamma(S_{1},p_{1*}L)$, where $p_{1*}L\cong Symm^{b}(V_{1}\oplus V_{2})$.
As a result of \mbox{Lemma \ref{LemSADSLB} },
\vspace{-18pt}

\begin{equation}
Symm^{b}(V_{1}\oplus V_{2})=\bigoplus_{\substack{i_{1},i_{2} \in \mathbb{N}^{2} \\ i_{1}+i_{2}=b}} \mathcal{O}_{S_{1}}(n_{1}H-\sum_{j=1}^{q^{2}+q+1} m_{1,j}B_{j})^{\otimes i_{1}} \otimes \mathcal{O}_{S_{1}}(n_{2}H-\sum_{j=1}^{q^{2}+q+1} m_{2,j}B_{j})^{\otimes i_{2}}.												
\end{equation}

\noindent Now, global sections of a direct sum of line bundles satisfy
\begin{equation}
\begin{split}
\Gamma(S_{1},\bigoplus_{\substack{i_{1},i_{2}\in\mathbb{N}^{2} \\ i_{1}+i_{2}=b}} \mathcal{O}_{S_{1}}(n_{1}H-\sum_{j=1}^{q^{2}+q+1} m_{1,j}B_{j})^{\otimes i_{1}} \otimes \mathcal{O}_{S_{1}}(n_{2}H-\sum_{j=1}^{q^{2}+q+1} m_{2,j}B_{j})^{\otimes i_{2}}) & \cong \\
\bigoplus_{\substack{i_{1},i_{2}\in\mathbb{N}^{2} \\ i_{1}+i_{2}=b}}\Gamma(S_{1},\mathcal{O}_{S_{1}}(n_{1}H-\sum_{j=1}^{q^{2}+q+1} m_{1,j}B_{j})^{\otimes i_{1}} \otimes \mathcal{O}_{S_{1}}(n_{2}H-\sum_{j=1}^{q^{2}+q+1} m_{2,j}B_{j})^{\otimes i_{2}}), & 
\end{split}
\end{equation}
and finally, we can conclude
\begin{equation}
\begin{split}
\bigoplus_{\substack{i_{1},i_{2}\in\mathbb{N}^{2} \\ i_{1}+i_{2}=b}}\Gamma(S_{1},\mathcal{O}_{S_{1}}(n_{1}H-\sum_{j=1}^{q^{2}+q+1} m_{1,j}B_{j})^{\otimes i_{1}} \otimes \mathcal{O}_{S_{1}}(n_{2}H-\sum_{j=1}^{q^{2}+q+1} m_{2,j}B_{j})^{\otimes i_{2}}) & \cong \\
\bigoplus_{\substack{i_{1},i_{2}\in\mathbb{N}^{2} \\ i_{1}+i_{2}=b}}\Gamma(S_{1},\mathcal{O}_{S_{1}}((i_{1}n_{1}+i_{2}n_{2})H-\sum_{j=1}^{q^{2}+q+1} (i_{1}m_{1,j}+i_{2}m_{2,j})B_{j})) & 
\end{split}
\end{equation}

\noindent Furthermore, by the hypothesis of the theorem, $\mathcal{O}_{S_{1}}((i_{1}n_{1}+i_{2}n_{2})H-\sum_{j=1}^{q^{2}+q+1} (i_{1}m_{1,j}+i_{2}m_{2,j})B_{j})$ is excellent with respect to the exceptional configuration of $\pi: S_{1}\rightarrow\mathbb{P}^{2}$ (see p. 215 in \cite{Harbourne85}), so if we denote $F=\mathcal{O}_{S_{1}}((i_{1}n_{1}+i_{2}n_{2})H-\sum_{j=1}^{q^{2}+q+1} (i_{1}m_{1,j}+i_{2}m_{2,j})B_{j})$, then by Theorem 1.1 in \cite{Harbourne85}:
\vspace{-12pt}

\begin{equation}
\begin{split}
h^{0}(S_{1}, F) & =\frac{1}{2}(f\cdot f-f\cdot k_{S_{1}})+1 \\
                & =\frac{1}{2}((i_{1}n_{1}+i_{2}n_{2})((i_{1}n_{1}+i_{2}n_{2})+3)-\sum_{j=1}^{q^{2}+q+1} (i_{1}m_{1,j}+i_{2}m_{2,j})((i_{1}m_{1,j}+i_{2}m_{2,j})+1))+1.
\end{split}
\end{equation}

\noindent This allows us to conclude that, when in the range wherein the bound on the minimum distance ensures the injectivity of the evaluation map, the dimension of the code is
\vspace{-12pt}

\begin{equation}
k=\sum_{\substack{i_{1},i_{2}\in\mathbb{N}^{2} \\ i_{1}+i_{2}=b}}\frac{1}{2}((i_{1}n_{1}+i_{2}n_{2})((i_{1}n_{1}+i_{2}n_{2})+3)-\sum_{j=1}^{q^{2}+q+1} (i_{1}m_{1,j}+i_{2}m_{2,j})((i_{1}m_{1,j}+i_{2}m_{2,j})+1))+1.
\end{equation}
\vspace{-6pt}

\noindent Finally, by Proposition 9.4 in \cite{EisenbudHarris16}, any section of $L\vert_{T_{1 B_{i}}}$ (recall that $T_{1 B_{i}}=p_{1}^{-1}(B_{i})$) intersects a fiber $C_{i}\subset T_{1 B_{i}}$ in a hyperplane, so $\mathit{l}=0$ and we can conclude the following bound for the minimum distance:
\begin{equation}
d\geq n-(q^{2}+q+1)(q+1)b.
\end{equation}
\end{proof}

\begin{Corollary}\label{CorCod2A4}
Let $S_{2}$ be the Deligne--Lusztig surface of type ${}^{2}A_{4}$ defined over the field $\mathbb{F}_{q^{2}}$. Consider $V_{i}=\pi^{*}i^{*}\mathcal{O}_{\mathbb{P}^{4}}(t_{i})$ for $i=1,2$, and let $T_{2}$ be the projective bundle $P(V_{1}\oplus V_{2})$ over $S_{2}$, $p_{2}: T_{2}\rightarrow S_{2}$. Then, for some $0<b<(q^{2}+1)$ and for any $t_{1},t_{2}\in\mathbb{N}$ such that $(q+1)<(i_{1}t_{1}+i_{2}t_{2})<(q^{3}+1)$ for any pair $i_{1},i_{2}\in\mathbb{N}$ with $i_{1}+i_{2}=b$, we can construct a code on $T_{2}$ over $\mathbb{F}_{q^{2}}$ with parameters
\begin{enumerate}
\item $n=(q^{5}+1)(q^{3}+1)(q^{2}+1)^{2}$;
\item $k=\sum_{\substack{i_{1},i_{2}\in\mathbb{N}^{2} \\ i_{1}+i_{2}=b}}\binom{4+i_{1}t_{1}+i_{2}t_{2}}{i_{1}t_{1}+i_{2}t_{2}}-\binom{4+i_{1}t_{1}+i_{2}t_{2}-(q+1)}{i_{1}t_{1}+i_{2}t_{2}-(q+1)}$;
\item $d\geq n-(q^{5}+1)(q^{3}+1)(q^{2}+1)b$.
\end{enumerate}
\end{Corollary}

\begin{proof}
Since all the rational points of $S_{2}$ are distributed equally on the disjoint rational curves $B_{i}$ constituting the irreducible components of the divisor $D_{2}\subset S_{2}$, then $\#S_{2}(\mathbb{F}_{q^{2}})=(q^{5}+1)(q^{3}+1)(q^{2}+1)$. Moreover, $\#\mathbb{P}^{1}(\mathbb{F}_{q^{2}})=(q^{2}+1)$, so we can conclude that the length $n$ of the code is
\begin{equation}
n=\#\mathcal{P}=\#S_{2}(\mathbb{F}_{q^{2}})(q^{2}+1)=(q^{5}+1)(q^{3}+1)(q^{2}+1)^{2}.
\end{equation}

\noindent Let $L$ be the line bundle $L=\mathcal{O}_{T_{2}}(b)$. In order to compute the dimension of the code, we have that $dim \Gamma(T_{2},L)=dim \Gamma(S_{2},p_{2*}L)$, with $p_{2*}L\cong Symm^{b}(V_{1}\oplus V_{2})$.
As a result of Lemma \ref{LemSADSLB} 
\begin{equation}
Symm^{b}(V_{1}\oplus V_{2})=\bigoplus_{\substack{i_{1},i_{2} \in \mathbb{N}^{2} \\ i_{1}+i_{2}=b}} \pi^{*}i^{*}\mathcal{O}_{\mathbb{P}^{4}}(t_{1})^{\otimes i_{1}} \otimes \pi^{*}i^{*}\mathcal{O}_{\mathbb{P}^{4}}(t_{2})^{\otimes i_{2}}.												
\end{equation}

\noindent Now, global sections of a direct sum of line bundles satisfy
\vspace{-12pt}
\begin{equation}
\Gamma(S_{2},\bigoplus_{\substack{i_{1},i_{2}\in\mathbb{N}^{2} \\ i_{1}+i_{2}=b}} \pi^{*}i^{*}\mathcal{O}_{\mathbb{P}^{4}}(t_{1})^{\otimes i_{1}} \otimes \pi^{*}i^{*}\mathcal{O}_{\mathbb{P}^{4}}(t_{2})^{\otimes i_{2}})\cong\bigoplus_{\substack{i_{1},i_{2}\in\mathbb{N}^{2} \\ i_{1}+i_{2}=b}}\Gamma(S_{2},\pi^{*}i^{*}\mathcal{O}_{\mathbb{P}^{4}}(t_{1})^{\otimes i_{1}} \otimes \pi^{*}i^{*}\mathcal{O}_{\mathbb{P}^{4}}(t_{2})^{\otimes i_{2}}),
\end{equation}
and since pull-back commutes with tensor product, it follows then that
\begin{align}
\bigoplus_{\substack{i_{1},i_{2}\in\mathbb{N}^{2} \\ i_{1}+i_{2}=b}}\Gamma(S_{2},\pi^{*}i^{*}\mathcal{O}_{\mathbb{P}^{4}}(t_{1})^{\otimes i_{1}} \otimes \pi^{*}i^{*}\mathcal{O}_{\mathbb{P}^{4}}(t_{2})^{\otimes i_{2}}) & \cong\bigoplus_{\substack{i_{1},i_{2}\in\mathbb{N}^{2} \\ i_{1}+i_{2}=b}}\Gamma(S_{2},\pi^{*}i^{*}\mathcal{O}_{\mathbb{P}^{4}}(i_{1}t_{1}+i_{2}t_{2}))
\end{align}

\noindent Furthermore, as $\pi$ is birational onto $Z$, we obtain (see Theorem 2.31 in \cite{Iitaka82})
\begin{equation}
\Gamma(S_{2},\pi^{*}i^{*}\mathcal{O}_{\mathbb{P}^{4}}(i_{1}t_{1}+i_{2}t_{2}))\cong\Gamma(Z,i^{*}\mathcal{O}_{\mathbb{P}^{4}}(i_{1}t_{1}+i_{2}t_{2}))
\end{equation}

\noindent Since $Z$ is the complete intersection of two hyper-surfaces $H_{0}, H_{1}\subset\mathbb{P}^{4}$ of degrees
$q+1$ and $q^{3}+1$, respectively, we obtain short exact sequences of line bundles
\begin{equation}
0\rightarrow\mathcal{O}_{\mathbb{P}^{4}}(t-(q+1))\rightarrow\mathcal{O}_{\mathbb{P}^{4}}(t)\rightarrow\mathcal{O}_{H_{0}}(t)\rightarrow 0,
\end{equation}
and
\begin{equation}
0\rightarrow\mathcal{O}_{H_{0}}(t-(q^{3}+1))\rightarrow\mathcal{O}_{H_{0}}(t)\rightarrow\mathcal{O}_{H_{0}\cap H_{1}}(t)\rightarrow 0,
\end{equation}
(see Section 7.3 in \cite{Iitaka82}). These sequences give long exact sequences of cohomology groups,
\begin{equation}
\begin{split}
0\rightarrow H^{0}(\mathbb{P}^{4},\mathcal{O}_{\mathbb{P}^{4}}(t-(q+1))) & \rightarrow H^{0}(\mathbb{P}^{4},\mathcal{O}_{\mathbb{P}^{4}}(t))\rightarrow H^{0}(H_{0},\mathcal{O}_{H_{0}}(t))\rightarrow \\
& H^{1}(\mathbb{P}^{4},\mathcal{O}_{\mathbb{P}^{4}}(t-(q+1)))\rightarrow H_{1}(\mathbb{P}^{4},\mathcal{O}_{\mathbb{P}^{4}}(t))\rightarrow \\
& H^{1}(H_{0},\mathcal{O}_{H_{0}}(t))\rightarrow H^{2}(\mathbb{P}^{4},\mathcal{O}_{\mathbb{P}^{N-1}}(t-(q+1)))\rightarrow \ldots,
\end{split}
\end{equation}
and
\begin{equation}
\begin{split}
0\rightarrow H^{0}(H_{0},\mathcal{O}_{H_{0}}(t-(q^{3}+1)) & \rightarrow H^{0}(H_{0},\mathcal{O}_{H_{0}}(t))\rightarrow \\
H^{0}(H_{0}\cap H_{1},\mathcal{O}_{H_{0}\cap H_{1}}(t))& \rightarrow H^{1}(H_{0},\mathcal{O}_{H_{0}}(t-(q^{3}+1))\rightarrow \ldots
\end{split}
\end{equation}

\noindent By the formulas of Theorem III.5.1 in \cite{Hartshorne77} for the cohomology of projective space, the first long exact sequence reduces (for any $t$) to
\begin{equation}
\begin{split}
0\rightarrow H^{0}(\mathbb{P}^{4},\mathcal{O}_{\mathbb{P}^{4}}(t-(q+1))) & \rightarrow H^{0}(\mathbb{P}^{4},\mathcal{O}_{\mathbb{P}^{4}}(t))\rightarrow \\
H^{0}(H_{0},\mathcal{O}_{H_{0}}(t))\rightarrow 0\rightarrow 0\rightarrow H^{1}(H_{0},\mathcal{O}_{H_{0}}(t))\rightarrow 0\rightarrow \ldots
\end{split}
\end{equation}

\noindent Hence, for any $t$, $H^{1}(H_{0},\mathcal{O}_{H_{0}}(t))=0$ and
\begin{align}
dim H^{0}(H_{0},\mathcal{O}_{H_{0}}(t)) & =dim H^{0}(\mathbb{P}^{4},\mathcal{O}_{\mathbb{P}^{4}}(t))-dim H^{0}(\mathbb{P}^{4},\mathcal{O}_{\mathbb{P}^{4}}(t-(q+1)) \\
& =\binom{4+t}{t}-\binom{4+t-(q+1)}{t-(q+1)}.
\end{align}

\noindent Consequently, for any $q+1<t<q^{3}+1$, the last sequence then gives
\begin{equation}
\Gamma(Z,\mathcal{O}_{Z}(t))\cong H^{0}(H_{0},\mathcal{O}_{H_{0}}(t)),
\end{equation}
so, when in the range wherein the bound on the minimum distance ensures the injectivity of the evaluation map, the dimension of the code is
\begin{equation}
k=\sum_{\substack{i_{1},i_{2}\in\mathbb{N}^{2} \\ i_{1}+i_{2}=b}}\binom{4+i_{1}t_{1}+i_{2}t_{2})}{i_{1}t_{1}+i_{2}t_{2})}-\binom{4+i_{1}t_{1}+i_{2}t_{2})-(q+1)}{i_{1}t_{1}+i_{2}t_{2})-(q+1)}.
\end{equation}

\noindent Finally, by Proposition 9.4 in \cite{EisenbudHarris16}, any section of $L\vert_{T_{2 B_{i}}}$ (recall that $T_{2 B_{i}}=p_{2}^{-1}(B_{i})$) intersects $C_{i}\subset T_{4 B_{i}}$ in an hyperplane, so $\mathit{l}=0$ and we can conclude the following bound for the minimum distance:
\begin{equation}
d\geq n-(q^{5}+1)(q^{3}+1)(q^{2}+1)b.
\end{equation}
\end{proof}

 Finally, we compute the parameters of some of the codes in Corollary \ref{CorCodA2} and Corollary \ref{CorCod2A4} for the binary case, $q=2$, as this is the most common finite field within applications.

\begin{Example}\label{ExmpA2}
Let us consider the particular case $q=2$, with $b=1$, $n_{1}=n_{2}=3$, $m_{1,j}=m_{2,j}=1$ for $j=1,2,3,$ and $m_{1,j}=m_{2,j}=0$ otherwise, for the family of codes presented in Corollary \ref{CorCodA2}. Then, we obtain codes with the following parameters:
\begin{enumerate}
\item $n=63$;
\item $k=14$;
\item $d\geq 42$.
\end{enumerate}
\end{Example}

\begin{Example}\label{Exmp2A4}
For $q=2$, let us consider the particular case $b=2$, $t_{1}=t_{2}=4$ for the family of codes presented in Corollary \ref{CorCod2A4}. Then we obtain a code with the following parameters:
\begin{enumerate}
\item $n=7425$,
\item $k=1107$,
\item and $d\geq 4455$.
\end{enumerate}

\end{Example}

\section{Conclusions}

By means of an intensive use of the intersection theory, we extend some of the previous results for codes over Deligne--Lusztig surfaces to the case of codes on projective bundles of rank $2$ over standard Deligne--Lusztig surfaces. In particular, we compute the length, dimension and give a lower bound for the minimum distance in Theorem \ref{ThmStdDLSur} for the cases of codes on projective bundles over Deligne--Lusztig surfaces of type $A_{2}(\mathbb{F}_{q}), {}^{2}A_{3}(\mathbb{F}_{q^{2}}), {}^{2}A_{4}(\mathbb{F}_{q^{2}})$ or $C_{2}(\mathbb{F}_{q})$. Moreover, in Corollary \ref{CorCodA2} and Corollary \ref{CorCod2A4}, we focus on some special families of these codes, by restricting our results to certain vector bundles over Deligne--Lusztig surfaces of type $A_{2}$ and ${}^{2}A_{4}$, in order to give a more explicit computation of their dimension. Finally, we give two examples of binary codes, as seen in Example \ref{ExmpA2} and Example \ref{Exmp2A4}, motivated by the fact that $q=2$ is the most common framework within applications. In Example \ref{ExmpA2}, we obtain an information rate $k/n=14/63$ and an error-correcting rate $\delta=d/n\geq 42/63$, and in Example \ref{Exmp2A4} we obtain an information rate $k/n=1107/7425$ and an error-correcting rate $\delta=d/n\geq 4455/7425$, concluding that both codes exhibit high error-correcting rates.

\color{black}
Some interesting problems would be to extend these results to the case of projective bundles of higher rank over standard Deligne--Lusztig surfaces following the techniques used in \cite{Nakashima06}, and to establish a more explicit understanding on how the morphism \\ $\pi: \overline{X}(w)\rightarrow Z$ of Theorem \ref{ThmHyp} behaves for Deligne--Lusztig varieties of dimension$\geq 3$ in order to construct new algebraic geometric error-correcting codes over them.


\end{document}